\documentclass[11pt]{amsart}
\usepackage[margin=3cm]{geometry}

\usepackage[utf8]{inputenc}
\usepackage{mathtools,amsfonts,amsthm,mathrsfs,amssymb}
\usepackage{textcomp}
\mathtoolsset{showonlyrefs}
\usepackage{bookmark}
\usepackage{microtype}
\usepackage{todonotes}
\usepackage{pgf,tikz,pgfplots}
\pgfplotsset{compat=1.15}
\usepackage{mathrsfs}
\usetikzlibrary{arrows}
\usepackage{thmtools}

\newtheorem{thm}{Theorem}
\newtheorem{lemma}[thm]{Lemma}

\newtheorem{corollary}[thm]{Corollary}

\theoremstyle{definition}
\newtheorem*{definition*}{Definition}

\newtheorem*{acknowledgement}{Acknowledgments}

\newcommand{\eps}{\varepsilon}

\newcommand{\be}{\mathbf{e}}
\newcommand{\tB}{\widetilde{B}}
\newcommand{\LOC}{\textsf{LOCAL}\;}

\DeclareMathOperator{\ad}{ad}
\DeclareMathOperator{\mad}{mad}

\renewcommand{\le}{\leqslant}

\renewcommand{\ge}{\geqslant}

\renewcommand{\epsilon}{\varepsilon}

\title{Distributed algorithms for fractional coloring}

\author[N. Bousquet]{Nicolas Bousquet}
\address{Laboratoire LIRIS (UMR 5205), CNRS, Université Claude Bernard Lyon 1, Univ. Lyon, France.}
\email{nicolas.bousquet@univ-lyon1.fr}

\author[L. Esperet]{Louis Esperet}
\address{Laboratoire G-SCOP, CNRS, Univ. Grenoble Alpes, Grenoble, France.}
\email{louis.esperet@grenoble-inp.fr}

\author[F.\ Pirot]{Fran\c{c}ois Pirot}
\address{Laboratoire G-SCOP, CNRS, Univ. Grenoble Alpes, Grenoble, France.}
\email{francois.pirot@grenoble-inp.fr}

\thanks{N.\ Bousquet, L.\ Esperet and F. Pirot are supported by ANR Projects GATO
(\textsc{ANR-16-CE40-0009-01}) and GrR (\textsc{ANR-18-CE40-0032}).} 

\begin{document}

\begin{abstract}
In this paper we study fractional coloring from the angle of
distributed computing. Fractional coloring is the linear relaxation of
the classical notion of coloring, and has many applications, in
particular in scheduling. It was proved by Hasemann, Hirvonen, Rybicki and
Suomela~\cite{HHRS16} that for every real $\alpha>1$
and integer $\Delta$, a fractional coloring of total weight at most $\alpha(\Delta+1)$
can be
obtained deterministically in a single round in graphs of maximum degree $\Delta$, in the \textsf{LOCAL}
model of computation.
However, a major issue of this result is that the output of each vertex has unbounded size.
Here we prove that even if we impose the more
realistic assumption that the output of each vertex has constant size,
we can find fractional colorings of total weight arbitrarily close to
known tight bounds for the fractional chromatic number in several
cases of interest.
More
precisely, we show that for any fixed $\eps > 0$ and $\Delta$,  a
fractional coloring of total weight at most 
  $\Delta+\eps$ can be
  found in $O(\log^*n)$ rounds in graphs of maximum degree $\Delta$
  with no $K_{\Delta+1}$, while finding a fractional coloring of total weight
  at most $\Delta$ in this case requires $\Omega(\log \log n)$ rounds for randomized
  algorithms and $\Omega( \log n)$ rounds for deterministic
  algorithms. We also show how to obtain fractional colorings of total
  weight at most $2+\epsilon$ in grids
  of any fixed dimension, for any
  $\epsilon>0$, in $O(\log^*n)$ rounds. Finally, we prove that in sparse
  graphs of large girth from any proper minor-closed family we can
  find a fractional coloring of total weight at most $2+\epsilon$, for any
  $\epsilon>0$, in $O(\log n)$ rounds.
\end{abstract}
\maketitle

\section{Introduction}

A \emph{(proper) $k$-coloring} of a graph $G$ is an assignment of colors to the
vertices of $G$, such that adjacent vertices receive different colors. This
is the same as a partition of the vertices of $G$ into (or covering of the vertices of $G$ by)
$k$ independent sets.
This has many applications in physical networks; for instance most scheduling problems can be expressed by a coloring problem in the underlying graph. 
When the resources at play inside the network are fractionable, it is more relevant to consider the linear relaxation
of this problem, where one wants to assign weights $x_S\in [0,1]$ to
the independent sets $S$ of $G$, so that for each vertex $v$ of $G$,
the sum of the weights $x_S$ of the independent sets $S$ containing
$v$ is at least $1$, and the objective is to minimize the sum of the
weights $x_S$. The solution of this linear program is the
\emph{fractional chromatic number} of $G$, denoted by $\chi_f(G)$. The
definition shows that $\chi_f$ is rational and that
$\omega(G)\le \chi_f(G)\le \chi(G)$ for any graph $G$ where
$\omega(G)$ denotes the clique number of $G$ (maximum size of a set of
pairwise adjacent vertices), and $\chi(G)$ denotes
the usual chromatic number (minimum $k$ such that $G$ admits a proper
$k$-coloring). A polyhedral
definition of $\chi_f$, which is not difficult to derive from the
definition above, is that $\chi_f(G)$ is
the minimum $x$ such that there is a probability distribution on the
independent sets of $G$, such that each vertex appears in a random
independent set (drawn from this probability distribution) with
probability at least $\tfrac1x$.

\smallskip

In this paper, we study fractional coloring from the angle of
distributed algorithms. In this context,  each
vertex outputs its ``part'' of the solution, and in a Locally
Checkable Labelling (whose precise definition will be given in Section~\ref{sec:prel}) this part should be of constant
size. For instance, in a distributed algorithm for proper
$k$-coloring of $G$, each vertex can output its color (an integer in $[k]=\{1,\ldots,k\}$), and
in a distributed algorithm for maximal independent set,
each vertex can output a bit saying whether it belongs to the independent
set. In both cases the fact
that the solution is correct can then be checked locally, in the sense
that adjacent vertices only need to compare their outputs, and if
there is no local conflict then the global solution is correct.
For more details about the distributed aspects of graph coloring,
the reader is referred to the book~\cite{BE13}. 

\smallskip

Looking back at the polyhedral definition of fractional
coloring introduced above, a first possibility would be to design a randomized distributed algorithm producing a
(random) independent set, in which each vertex has a large probability
to be selected (in this case, the output of each vertex is a single
bit, telling whether it belongs to the chosen independent set). A
classical algorithm in this vein is the following \cite{AS,BHR18}: Each vertex is assigned a
random identifier, and joins the independent set if its identifier is smaller than
that of all its neighbors. This clearly
produces an independent set, and it is not difficult to prove that each
vertex $v$ is selected with probability at least $\tfrac1{d(v)+1}$, so in particular this 1-round randomized algorithm witnesses
the fact that the fractional chromatic number of graphs of maximum
degree $\Delta$ is at most $\Delta+1$. Note that \emph{factor-of-IID} algorithms
for independent sets introduced in the past years are of this form
(see for instance~\cite{CGHV15,CHV18}).
This leaves the question of
how to produce a \emph{deterministic} distributed algorithm for
fractional coloring. Recall that a fractional coloring is a distribution of
independent sets, so the first issue is to decide what the output of
the algorithm should be in order to be locally checkable.
A solution explored in~\cite{HHRS16} is to assign
to each independent set $S$ of $G$ an interval $I_S\subset \mathbb{R}$
of length $x_S$ (or a finite union of intervals of total length
$x_S$), where $x_S$ is the weight defined above in the linear programming definition of
fractional coloring, such that all $I_S$ are pairwise disjoint. The
output of each vertex $v$ is then the union of all subsets $I_S\subset
\mathbb{R}$ such
that $v\in S$. Each vertex can check that its
output (which is a finite union of intervals) has total length at least 1, and pairs of adjacent vertices can check that
their outputs are disjoint, so the fact that this is a fractional
coloring can be checked locally. If the set of identifiers of the vertices of
$G$ is
known to all the vertices in advance (for instance if there are $n$ vertices and
the identifiers are $\{1,\ldots,n\}$), then by enumerating all
permutations of the identifiers in some canonical order, it is not difficult to transform the 1-round
randomized algorithm described above into a 1-round deterministic
algorithm producing such an output and with
total weight at most $\Delta+1$ (by running it for all permutations and
aggregating all the solutions).
The main result of~\cite{HHRS16} is a
1-round deterministic algorithm producing such an output and with
total weight at most $\alpha(\Delta+1)$ (for any $\alpha>1$), when the set of identifiers is not known in advance. As observed in~\cite{HHRS16}, the
unbounded size of the output implies that this algorithm is unusable
in practice. In this paper, we explore a different way to design
deterministic distributed algorithms producing fractional colorings of
small total weight.

\medskip

Given two integers $p \ge q \ge 1$, a \emph{$(p\!:\!q)$-coloring} of a
graph $G$ is an assignment of $q$-element subsets of $[p]$ to the
vertices of $G$, such that the sets assigned to any two adjacent
vertices are disjoint. An alternative view is that a
$(p\!:\!q)$-coloring of $G$ is precisely a homomorphism from $G$ to
the \emph{Kneser graph} $\mathrm{KG}(p,q)$, which is the graph whose vertices
are the $q$-element subsets of $[p]$, and in which two vertices are
adjacent if the corresponding subsets of $[p]$ are disjoint. The \emph{weight} of a $(p\!:\!q)$-coloring $c$ is $w(c) = p/q$.

\smallskip

The fractional chromatic number $\chi_f(G)$ can be equivalently
defined as the infimum of $\{\tfrac{p}{q}\,|\, G \mbox{ has a }
(p\!:\!q)\mbox{-coloring}\}$~\cite{SU13} (as before, it can be proved that this
infimum is indeed a minimum). Observe that a
$(p\!:\! 1)$-coloring is a (proper) $p$-coloring. It is well known that
the Kneser graph $\mathrm{KG}(p,q)$ has fractional chromatic number
$\tfrac{p}q$, while Lov\'asz famously proved~\cite{L78} that its
chromatic number is $p-2q+2$ using topological methods. This shows in
particular that $\chi_f(G)$ and $\chi(G)$ can be arbitrarily far apart.

This definition of the fractional chromatic number gives a natural way
to produce distributed fractional colorings, while keeping the output
of each vertex bounded. It suffices to fix the integer $q\ge 1$, and ask
for the smallest integer $p$ such that $G$ has a
$(p\!:\!q)$-coloring $c$; then the output of each vertex is the
sequence of its $q$ colors from $[p]$, which can be encoded in a
bit-string with at most $q\log
p = q(\log q + \log w(c))$ bits (in the remainder of the paper,
the \emph{output size} always refers to the number of bits in this
string, and $\log$ stands for the binary logarithm). The requirement that the output of each vertex has bounded size
is quite constraining in the case of fractional colourings, since in general the smallest
integers $p$ and $q$ such that $\chi_f(G)=\tfrac{p}q$  can be exponential in $|V(G)|$
\cite{Fis95}.


\medskip

In addition to the classical applications of fractional coloring in
scheduling (see~\cite{HHRS16}), there is another more theoretical motivation for
studying this problem (or the relaxation above where $q$ is fixed). In
 $n$-vertex graphs of maximum degree $\Delta$ (which is assumed to be a constant),
 a coloring with $\Delta+1$ colors can be found in
$O(\log^*n)$ rounds~\cite{GPS88,L92} in the \textsf{LOCAL} model of
computation (which will be introduced formally below). On the other
hand, Brooks' Theorem says that if $\Delta\ge 3$, any graph of maximum
degree $\Delta$ with no
clique $K_{\Delta+1}$ is $\Delta$-colorable~\cite{Bro41}, and finding such a
coloring has proved to be an interesting problem of
\emph{intermediate} complexity in distributed computing. It was proved
that the round complexity for computing such a coloring is $\Omega(\log
  \log n)$ for randomized algorithms~\cite{BFHKLRSU16} and $\Omega(\log n)$ for
deterministic algorithms~\cite{CKP19} (see also~\cite{Bra19}). Thus a large complexity gap
appears between $\Delta$ and $\Delta+1$ colors, and since the values
are integral it is all that can be said about this problem. However,
if the number of colors is real, or rational, the precise location of the complexity threshold
in the interval $[\Delta,\Delta+1]$ can be investigated. In Section~\ref{sec:mdeg}, we will
  show that for fractional coloring, the complexity threshold is
  arbitrarily close to $\Delta$; namely finding a fractional
  $\Delta$-coloring is as difficult as finding a $\Delta$-coloring, but
  for any fixed $\eps > 0$ and $\Delta$,  a fractional
  $(\Delta+\eps)$-coloring with output size of
  $O\left(\frac{1}{\eps}\log \frac{\Delta}{\eps}\right)$ bits per
  vertex can be
  found in $O_\epsilon(\log^*n)$ rounds in graphs of maximum degree $\Delta$
  with no $K_{\Delta+1}$. Here, the subscript $\epsilon$ in the big-$O$
  notation indicates that the implicit multiplicative constant depends
  on $\epsilon$.
  
\begin{restatable}{thm}{thmdeltafrac}\label{thm:deltafrac}
For any integer $q\ge 1$, and any $n$-vertex graph $G$ of maximum
degree $\Delta\ge 3$, without $K_{\Delta+1}$, a
$(q\Delta+1\!:\!q)$-coloring of $G$ can be computed in
$O(q^3\Delta^{2q}+q\log^*n)$ rounds deterministically  in the \LOC
model. 
\end{restatable}


  There are other similar complexity thresholds in distributed graph
  coloring. For instance, it was proved that $D$-dimensional
  grids\footnote{We note that these results are proved for toroidal
    grids with a consistent orientation, while
    Theorem~\ref{thm:gridfrac} considers classical, non-oriented grids.} can
  be colored with 4 colors in $O(\log^*n)$ rounds, while computing a 3-coloring in a
  2-dimensional $n\times n$-grid takes $\Omega(n)$
  rounds~\cite{BHKLOPRS17} (see also \cite{HSW17} for related results).
For (almost) vertex-transitive graphs like grids finding minimum fractional colorings
is essentially equivalent to finding maximum independent sets, and
simple local randomized algorithms approaching the optimal independent
set in grids can be used to produce fractional
$(2+\epsilon)$-colorings with small output
(see for instance~\cite{CHV18}).
  In
  Section~\ref{sec:grid}, we will show that for any fixed $\eps > 0$ and $D\ge
  1$, a
  fractional $(2+\eps)$-coloring  of the $D$-dimensional grid $G(n,D)$ of
  dimension $n\times \cdots \times n$ with output size of $O(\tfrac{6^D}{\eps}\log(\tfrac{6^D}{\eps} ))$
  bits per vertex can
  be computed \emph{deterministically} in $O_{\eps,D}(\log^* n)$ rounds, while it can be easily observed
  that finding a $(2q\!:\! q)$-coloring takes $\Omega(n)$ rounds (even
  if $D=1$, i.e., when the graph is a path). 
  
  \begin{restatable}{thm}{thmgridfrac}\label{thm:gridfrac}
    For every integers $D\ge 1$ and $q\ge 1$, a
$(2q+4\cdot 6^D\!:\! q)$-coloring of the $D$-dimensional grid $G(n,D)$ can
be found in $O\big(D\ell(2\ell)^D+D\ell\log^*n\big)$ rounds deterministically  in the \LOC
model, where $\ell = q+2\cdot 6^D$.
\end{restatable}


The last observation
  implies in particular that  $(2q\!:\! q)$-coloring trees takes $\Omega(n)$
  rounds. On the other hand, trees can be colored with 3 colors in
  $O(\log n)$ rounds and this is best possible (even with 3 replaced
  by an arbitrary number of colors)~\cite{GPS88,L92}. The
  \emph{maximum average degree} of a graph $G$, denoted by $\mad(G)$,
  is the maximum of the average degrees of all the subgraphs $H$ of $G$.  In
  Section~\ref{sec:sparse} we prove that, for every $\eps>0$, graphs of maximum average
  degree at most $2+\eps/40$ and large girth can be  $(2+\eps)$-colored in
  $O_\eps(\log n)$ rounds with output size of $O(\tfrac1{\eps}\log
  \tfrac1{\eps})$ bits per vertex. 

 \begin{restatable}{thm}{thmdistmad}\label{thm:distrmad2}
Let $G$ be an $n$-vertex graph with girth at least $2q+2$, and $\mad(G)\le 2+\tfrac1{40q}$, for some
  fixed $q\ge 1$. Then a $(2q+1\!:\!q)$-coloring of $G$ can be
  computed deterministically  in
  $O(q\log n+q^2)$ rounds in the \LOC
model.
\end{restatable}


  This implies that trees, and more generally
  graphs of sufficiently large girth from any minor-closed class can be  $(2+\eps)$-colored in
  $O(\log n)$ rounds, for any $\eps>0$. Note that the assumption that
  the girth is large cannot be avoided, as a cycle of length $2q-1$
  has fractional chromatic number equal to
  $2+\tfrac1{q-1}>2+\tfrac1q$.





  We conjecture that more generally, graphs of girth $\Omega(q)$ and
  maximum average degree $k+O(1/q)$ (where $k\ge 2$ and $q$ are integers) have a $(kq+1\!:\!q)$-coloring
   that can be computed efficiently by a  deterministic algorithm in
   the \LOC model (the fact that such a coloring exists is a simple
   consequence of~\cite{NS19} but the proof there uses flows and does not seem
   to be efficiently implementable in the \LOC model).

  \section{The \LOC model of computation}\label{sec:prel}

 All our results are
  proved in the \LOC model, introduced by Linial~\cite{L92}. We
  consider a network, in the form of an $n$-vertex graph $G$ whose vertices have
  unbounded computational power, and whose edges are communication links
  between the corresponding vertices. We are given a combinatorial
  problem that we need to solve in the graph $G$. In the
  case of deterministic algorithms, each
  vertex of $G$ starts with an arbitrary unique identifier (an integer between
  1 and $n^c$, for some constant $c\ge 1$, such that all integers
  assigned to the vertices are distinct). For randomized algorithms, each vertex
  starts instead with a collection of (private) random bits. The vertices then exchange messages
  (of unbounded size)
  with their neighbors in synchronous rounds, and after a fixed number
  of rounds (the \emph{round complexity of the algorithm}), each
    vertex outputs its local ``part'' of the global solution of the
    problem. This could for instance be the color of the vertex in a
    proper $k$-coloring. In Locally Checkable Labelling (\textsf{LCL})
    problems, this output has to be of constant size, and should be
    checkable locally, in the sense that the solution is correct
    globally if
    and only if it is correct in
    all neighborhoods of some (constant) radius. \textsf{LCL} problems include
    problems like $k$-coloring (with constant $k$), or maximal
    independent set, but not maximum independent set (for instance),
    and are central in the field of distributed algorithms.

    It turns out that with the assumption that
    messages have unbounded size, vertices can just send to their
    neighbors at each round all the information that they have
    received so far, and in $t$ rounds each vertex $v$ ``knows'' its
    neighborhood $B_t(v)$ at distance $t$ (the set of all vertices at
    distance at most $t$ from $v$). More specifically $v$ knows the
    labelled subgraph of $G$ induced by $B_t(v)$ (where the labels are
    the identifiers of the vertices), and nothing more, and the
    output of $v$ is based solely on this information.

    The goal is to minimize the round complexity. Since in $t$ rounds
    each vertex sees its neighborhood at distance $t$, after a
    number of rounds equal to the diameter of $G$, each vertex
    sees the entire graph. Since each vertex has
    unbounded computational power, a distinguished vertex (the vertex
    with the smallest identifier, say) can
    compute an optimal solution of the problem and communicate this
    solution to all the vertices of the graph. This shows that any
    problem can be solved in a number of rounds equal to the diameter
    of the graph, which is at most $n$ when $G$ is connected.  The
    goal is to obtain algorithms that are significantly more
    efficient, i.e., of round complexity $O(\log n)$, or even
    $O(\log^* n)$, where $\log^* n$ is the number of times we have to
    iterate the logarithm, starting with $n$, to reach a value in $(0,1]$.






\section{Maximum degree}\label{sec:mdeg}

In this section we will need the following consequence of a result of Aubry,
Godin and Togni~\cite[Corollary 8]{AGT14} (see
also~\cite{CGHHJ00}).


\begin{thm}[\cite{AGT14}]\label{thm:pathchoos}
Let $q\ge 1$ be an integer and let $P=v_1,v_2,\ldots,v_{2q+1}$ be a
path. Assume that for $i\in\{1,2q+1\}$ the vertex $v_i$ has a list
$L(v_i)$ of at least $q+1$ colors, and for any $2\le i \le 2q$, $v_i$
has a list $L(v_i)$ of at least $2q+1$ colors. Then each vertex $v_i$ of $P$ can be assigned a
subset $S_i\subseteq L(v_i)$ of $q$ colors, so that adjacent vertices
are assigned disjoint sets.
\end{thm}

In a graph $G$, we say that a path $P$ is an \emph{induced path} if
the subgraph of $G$ induced by $V(P)$, the vertex set of $P$, is a
path. The \emph{length} of a path is its number of vertices. Note that shortest paths are induced paths, and in particular every connected
graph that has no induced path of length $k$ has diameter at most $k-2$.


We are now ready to prove Theorem~\ref{thm:deltafrac}, which we
restate here for the convenience of the reader.

\thmdeltafrac*


\begin{proof}
The first step is to construct the graph $H_{1}$, whose nodes
are all the induced paths of length $2q+1$ 
in $G$, with an edge between two nodes of $H_1$ if
the corresponding paths in $G$ share at least one vertex. So $H_1$ can
be seen as the intersection graph of the induced paths of length
$2q+1$ of $G$. Note that any communication in $H_{1}$ can
be emulated in $G$, by incurring a multiplicative factor of $O(q)$
on the round complexity. Note that $H_1$ has at most $n\cdot \Delta^{2q}$
vertices and maximum degree $O(q^2 \Delta^{2q})$ (given a path $P$, one has $O(q^2)$ possible choices for the position of the intersection vertex $x$ on $P$ and on an intersecting path, then at most $\Delta^{2q}$ possible ways of extending $x$ into such an intersecting path).
It follows that a maximal independent set
$S_1$ in $H_{1}$ can be computed in $O(\Delta(H_1)+\log^*(|V(H_1)|))=O(q^2 \Delta^{2q}+\log^*(\Delta^{2q} n))=O(q^2 \Delta^{2q}+\log^*n)$ rounds in $G_{1}$~\cite{BEK14},
and thus in $O(q^3 \Delta^{2q}+q\log^*n)$ rounds in $G$.

\smallskip

Observe that the set $S_1$ corresponds to a set of vertex-disjoint
induced paths of length $2q+1$ in $G$. Let $\mathcal{P}=\bigcup_{P\in
  S_1} V(P)$. By maximality of $S_1$, the graph $G-\mathcal{P}$ has no induced
path of length $2q+1$, and in particular each connected component $C$ of
$G-\mathcal{P}$ has diameter at most $2q-1$. Each such component $C$ has indeed a
$(q\Delta+1\!:\!q)$-coloring  $c$ (since $C$, as a subgraph of $G$, is $\Delta$-colorable, 
by Brooks' theorem), which can be computed in $O(q)$ rounds. Our next step
is to extend this coloring $c$ of $G-\mathcal{P}$ to $\mathcal{P}$. To
this end, define a graph $H_2$ whose nodes are the elements of $S_1$, in which two
nodes are adjacent if the two corresponding paths are adjacent in $G$
(i.e., some edge of $G$ has a vertex in each of these two paths). Observe
that $H_2$ has at most $n$ nodes and maximum degree $O(q \Delta)$.
So a proper coloring $c_2$ of $H_2$ with $t=O(q \Delta)$ colors
$1,2,\ldots,t$ can be found in $O(q\Delta+\log^*n)$ rounds in $H_2$~\cite{BEK14}
(and thus in $O(q^2\Delta+q\log^*n)$ rounds in $G$). For each color
$i$, consider the paths of $\mathcal{P}$ of color $i$ in
$c_2$. We extend the current partial coloring of $G$ to these paths
(which are pairwise non-adjacent by definition) using Theorem
\ref{thm:pathchoos}. For each of these paths, each vertex starts with
$q\Delta+1$ available colors and the coloring of the neighborhood of
this path forbids at most $q(\Delta-2)$ colors for each internal
vertex of the path, and at most $q(\Delta-1)$ colors for each
endpoint of the path. Thus each internal vertex of a path of
$\mathcal{P}$ has a list of $2q+1$ colors and each endpoint has a list
of $q+1$ colors, as required for the application of Theorem
\ref{thm:pathchoos}. The extension thus takes $t=O(q \Delta)$ steps,
each taking $O(q)$ rounds, and thus the final round complexity is $O(q^3 \Delta^{2q}+q\log^*n)$.
\end{proof}

We now prove that finding a $(q\Delta\!:\!q)$-coloring is
significantly harder. We will use a reduction from the \emph{sinkless
  orientation} problem: given a bipartite $n$-vertex $\Delta$-regular graph
$G$, with $\Delta\ge 3$ we have to find an
orientation of the edges of $G$ so that each vertex has at least one
outgoing edge. It was proved that in the \LOC model this takes $\Omega(\log \log n)$
rounds for a randomized algorithm~\cite{BFHKLRSU16} (see
also~\cite{CKP19}) and $\Omega(\log n)$ for a
deterministic algorithm~\cite{CKP19} (see also~\cite{Bra19}). Note
that the results in~\cite{BFHKLRSU16}  and~\cite{Bra19} are proved for
$\Delta=3$, while the results of~\cite{CKP19} are proved 
for any $\Delta\ge 3$.

\begin{figure}[htb]
 \centering
 \includegraphics{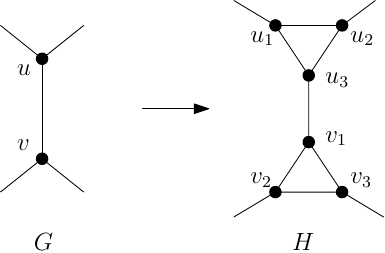}
 \caption{The construction of $H$ from $G$ in the proof of
   Theorem~\ref{thm:lbfrac}, for $\Delta=3$.}
 \label{fig:deltawy}
\end{figure}

\begin{thm}\label{thm:lbfrac}
For any integers $\Delta\ge 3$ and  $q\ge 1$, obtaining a $(q\Delta\!:\! q)$-coloring of an $n$-vertex
$\Delta$-regular graph with no $K_{\Delta+1}$ takes $\Omega(\log \log n)$
rounds for a randomized algorithm and $\Omega(\log n)$ rounds for a
deterministic algorithm in the \textsf{LOCAL} model.
\end{thm}

\begin{proof}
Let $\mathcal{A}$ be a distributed algorithm which returns a $(q\Delta:q)$-coloring  of any $n$-vertex $\Delta$-regular graph $G$ with no $K_{\Delta+1}$ within $f(n)$ rounds, for every integer $q\ge 1$. Note that such an algorithm exists with $f(n)\le n$, since by Brook's theorem $G$ is $\Delta$-colorable, and we obtain a $(q\Delta:q)$-coloring of $G$ by exploding each color into $q$ copies in a $\Delta$-coloring of $G$ (that we find for free once each node has a complete knowledge of the graph $G$ within $n$ rounds). 

Consider a bipartite $\Delta$-regular $n$-vertex graph $G$, in which we want
to compute a sinkless orientation. Let $H$ be the $\Delta n$-vertex $\Delta$-regular graph obtained
from $G$ by replacing each vertex $v$ by a clique
$v_1,v_2,\ldots,v_\Delta$ on $\Delta$ vertices (see
Figure~\ref{fig:deltawy}). Note that $H$ does not contain any copy of
$K_{\Delta+1}$. 
We now apply the algorithm $\mathcal{A}$ on $H$ in order to compute a $(q\Delta:q)$-coloring $c$ of $H$.
Note that each copy of $K_\Delta$ in $H$ uses all $q\Delta$
colors, and thus the set $S$ of vertices of $H$ whose set of colors in $c$ 
contains the color 1 intersects each copy $v_1,\ldots,v_\Delta$ of
$K_\Delta$ in $H$ in
a single vertex (say $v_1$, up to renaming of the vertices of
$H$). Let $e_v$ be the edge of $H$ that contains $v_1$, but is
disjoint from $v_2,\ldots,v_\Delta$. Then for every vertex $v$ of $G$ we
orient $e_v$ from $v$ to the other endpoint of $e_v$ in $H$. Note that this
gives a partial orientation of $H$ (an edge cannot be oriented
in both directions, since otherwise the two endpoints contain color 1,
which is a contradiction), which can be transferred to a partial
orientation of $G$ by contracting each clique $v_1,\ldots,v_\Delta$ back to
the vertex $v$. Since $S$ intersects each copy  of $K_\Delta$ in $H$, the
resulting partial
orientation of $G$ is sinkless, as desired.

We have described a distributed algorithm which returns a sinkless orientation of any bipartite $\Delta$-regular graph $n$-vertex graph within $f(\Delta n)$ rounds, hence $f(n) = \Omega\left(\log \frac{n}{\Delta}\right)$ if $\mathcal{A}$ is deterministic, or $f(n) = \Omega\left(\log \log \frac{n}{\Delta}\right)$ if $\mathcal{A}$ is randomized.
\end{proof}

A natural question is whether one can find another description of
fractional coloring with bounded size certificates for which the
existence of a fractional coloring of total weight $\Delta$ can be
determined faster. The existence of such a fractional coloring implies
the existence of an independent set of weight at least $\frac 1
\Delta$ in $G$. It can be observed that the same proof as that of
Theorem~\ref{thm:lbfrac}, shows that finding such an independent set
in a graph of maximum $\Delta$ with no $K_{\Delta+1}$ is as hard as finding a sinkless orientation.

\section{Coloring the grid}\label{sec:grid}

Let $G(n,D)$ be the $D$-dimensional grid with vertex set $[n]^D$,
i.e., two distinct vertices $x=(x_1,\ldots,x_D)$ and
$y=(y_1,\ldots,y_D)$ in $[n]^D$ are adjacent in $G(n,D)$ if and only if
$d_1(x,y)=\sum_{i=1}^D |x_i- y_i|=1$ (where $d_1$ denotes the usual
taxicab distance). We assume that all the vertices know their identifier 
(but do not have access to their own coordinates in the grid). We note
that our results do not assume any knowledge of directions in the grid
(i.e., a consistent orientation of edges, such as South$\to$North and West$\to$East in 2 dimensions), in contrast with the
results of~\cite{BHKLOPRS17}.

Note that the distance between vertices in the grid coincides with the
taxicab distance $d_1$ in $\mathbb{R}^D$. 
In this section it will be convenient to work instead with the Chebyshev 
distance $d_\infty(x,y)=\max_{1\le i \le n}|x_i-y_i|$, since balls with respect to $d_\infty$ are grids themselves, as well as their (possibly empty) pairwise intersections.
We observe that any communication between two nodes at distance $d_\infty$ at most $\ell$ can be emulated within $D\cdot \ell$ rounds, hence working with $d_\infty$ rather than with $d_1$ does not have a significant impact in terms of round complexity.
In the remainder of this section, all distances refer to the $d_\infty$-distance.


We now prove  Theorem~\ref{thm:gridfrac}, restated here for convenience.

\thmgridfrac*


\begin{proof}
Let $G=G(n,D)$ and  $\ell = q+2\cdot 6^D$. We assume that $n\ge 2\ell$, for otherwise $G$ has diameter (with respect to $d_\infty$) at most
$O(\ell)$ and a desired coloring can be found in $O(D\ell)$ rounds. We start by finding an
inclusionwise maximal independent set $I$ of $G^{[\ell]}$, the graph obtained from $G$ by linking each pair of vertices at ($d_\infty$-)distance at most $\ell$ from each other, which are therefore linked by a path of length at most $D\ell$ in $G$. Note that the maximum degree of $G^{[\ell]}$ corresponds to the maximum size of a ball of radius $\ell$ in a $D$-dimensional grid, that is at most $(2\ell+1)^D$; therefore $I$ can be constructed in $O((2\ell+1)^D+\log^* n)$ rounds in $G^{[\ell]}$~\cite{BEK14}, and this can be emulated in $O\big(D\ell(2\ell)^D+D\ell\log^*n\big)$ rounds in $G$.


A given vertex $x$ of the grid can find the list $L_x$ of vertices of
coordinates $x+\be$ for every $\be \in \{-1,1\}^D$ within $D$ rounds (although $x$ is not aware of the absolute directions in the grid corresponding to each of these vertices).
This can be done with the following procedure. Given a vertex $y\in
V(G)$, we say that a vertex $z\in V(G)$ is a $1$-neighbour of $y$ if
$N(z)$ contains $y$, and for every $2\le i \le D$ we say that $z$  is an $i$-neighbour of $y$ if $N(z)$ contains at least two $(i-1)$-neighbours of $y$. Then $L_x$ is the list of vertices $y\in V(G)$ such that $x$ is a $D$-neighbour of $y$, and this list can be found within $D$ rounds.
Given $x$ and $y=x+\be \in L_x$, finding $z=x+2\be$ can be done in
$2D$ rounds, since $z$ is the vertex in $L_y$ furthest away from $x$
with respect to the $d_1$-distance in $G$ (and more generally $x+i\be$ can be found in
$iD$ rounds).

The next step of the coloring procedure is as follows. Every vertex
$x\in I$ chooses a direction $\be_x\in  \{-1,1\}^D$ in such a way that $x[i]\coloneqq x+i \cdot \be_x\in
[n]^D$ is well-defined for every $i\le \ell$ (this is possible since $n\ge 2\ell$, and
such a direction can be chosen in $\ell\cdot D$ rounds).

For any $x\in I$ and $1\le i \le \ell$, we define a set $B(x,i)$ as
follows. Each vertex $y$ considers the 
vertex $x\in I$ of smallest identifier such that $d_\infty(y,x[i])\le 2\ell$
and joins the set $B(x,i)$. Equivalently, $B(x,i)$ is the ball of
center $x[i]$ and radius $2\ell$, in which we remove all vertices at
distance at most $2\ell$ from $x'[i]$ for some $x'\in I$ of smaller identifier
than that of $x$. We also let $\tB(x,i)$ be obtained from $B(x,i)$ after removing all vertices at distance exactly $2\ell$ from $x[i]$. When some vertex $v$ is in $B(x,i)\setminus \tB(x,i)$, we say that $v$ is a \emph{boundary vertex for $x[i]$}. An example of such a partition of the grid is depicted in Figure~\ref{fig:grid-partition}.

Each vertex $v\in V(G)$ is at distance at most $\ell$ from at least one vertex $x\in I$ by maximality of $I$, and therefore at distance at most $2\ell$ from $x[i]$ for every $i\le \ell$. It follows that for every fixed $1\le i \le \ell$, the collection of the sets $B(x,i)$ over all $x\in I$ forms a partition of $V(G)$. 

We now show that for each $i$, no two distinct sets $\tB(x,i)$ and $\tB(y,i)$
(with $x,y\in I$) are connected by an edge. Indeed, assume for the
sake of contradiction that there is an edge $uv$ with $u\in \tB(x,i)$
and $v\in 
\tB(y,i)$, for two distinct vertices $x$ and $y$ in $I$ with
$\text{ID}(x)<\text{ID}(y)$, where $\text{ID}(x)$ denotes the
identifier of $x$. Then $d_\infty(x,u)\le 2\ell-1$ and thus
$d_\infty(x,v)\le 2\ell$ which contradicts the fact that $v\not\in
B(x,i)$ since $\text{ID}(x)<\text{ID}(y)$. This shows that no two distinct sets $\tB(x,i)$ and $\tB(y,i)$
(with $x,y\in I$) are connected by an edge, as desired. This implies
that all the components of the subgraph $G_i$ of $G$ induced by $\bigcup_{x\in
  I}\tB(x,i)$ (which is bipartite) have diameter
 $O(\ell)$, and in particular for every $1\le i\le \ell$, we can find a proper $2$-coloring $c_i$ of $G_i$
with colors in $\{2i-1,2i\}$ within $O(\ell)$ rounds. We now show that the union of these colorings over all $i\le \ell$ yields a $(2q+4\cdot 6^D:q)$-coloring $c$ of $G$.

The total number of colors is $2\ell = 2q+4\cdot 6^D$, so
it remains to show that each vertex $v\in V(G)$ is assigned at
least $q$ colors in $c$. This is equivalent to showing that each
vertex $v\in V(G)$ is a boundary vertex for $x[i]$ for at most $2\cdot
6^D$ different combinations of $x$ and $i$. If $v$ is a boundary
vertex for $x[i]$, then $x[i]$ lies on the boundary of the ball $B_v$
of center $v$ and radius $2\ell$. Note that every 
line intersects the boundary of a convex polytope in at most two points or in a
segment, and in the latter case the line is contained in the hyperplane defining a
facet. Since we have chosen the directions $\be_x \in
\{-1,1\}^D$ while balls in $d_\infty$ are grids (bounded by
axis-parallel hyperplanes), this shows that for every vertex $x\in I$, the set
of vertices $\{x[i] : 1\le i\le \ell\}$ intersects the boundary of
$B_v$ at most twice, and if the intersection is non-empty then $x$ is
at distance at most $3\ell$ from $v$. For a fixed vertex $v$ there can
be at most $6^D$ vertices in $N_{G^{[3\ell]}}(v) \cap I$. To see this,
for every $(i_1,\ldots,i_D) \in \{-3,-2,\ldots,2\}^D$, we let
$S_{(i_1,\ldots,i_D)}$ be the set of vertices $y \in V(G)$ of
coordinates $(y_1,\ldots,y_D)$ satisfying $v_j + i_j\cdot \ell \le y_j
\le v_j + (i_j+1)\cdot \ell$ for every $1\le j \le D$. It is
straightforward to see that the diameter of $S_{(i_1,\ldots,i_D)}$ is $\ell$, so it contains at most 1 element of $I$. Since moreover the collection $\big(S_{(i_1,\ldots,i_D)}\big)$ covers $N_{G^{[3\ell]}}(v)$, the results follows, which concludes the proof.
\end{proof}


\begin{figure}
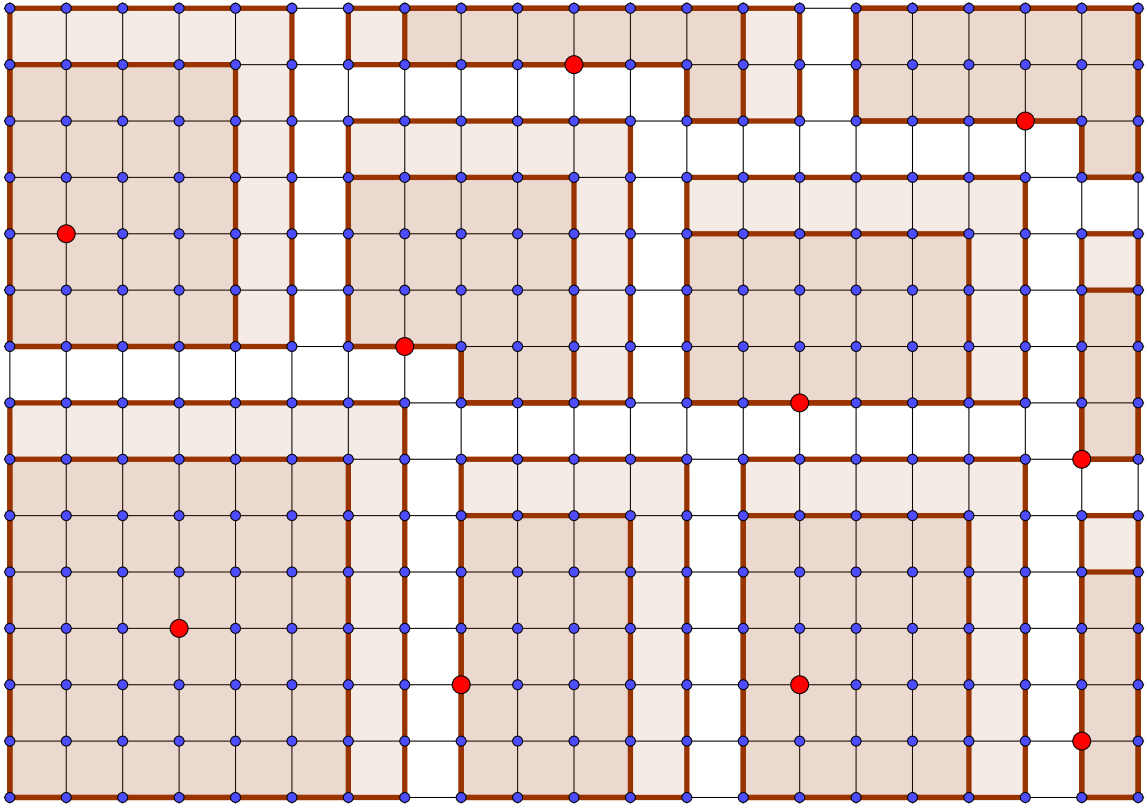

\include{grid_partition}
\caption{An example of the partition of the $2$-dimensional grid into the sets $B(x,i)$; here $2\ell=4$, and the labels of the vertices are increasing according to the lexicographical ordering of the coordinates. The regions containing the boundary vertices are lighter; the vertices $x[i]$ (for $x\in I$) are colored in red.}
\label{fig:grid-partition}
\end{figure}

\section{Sparse graphs}\label{sec:sparse}

The \emph{average degree} of a graph $G=(V,E)$, denoted by $\ad(G)$, is defined as the average of the
degrees of the vertices of $G$ (it is equal to 0 if $V$ is empty and to $2|E|/|V|$ otherwise).

In this section we are interested in graphs of  average degree
at most $2+\epsilon$, for some small $\epsilon>0$, and with no
connected component isomorphic to a short cycle (later we will need
the stronger property that all cycles in the graph are large). We first prove that
they contain a linear number of vertices that are either of degree at
most 1 or belong to long chains
of vertices of degree 2, and such a set can be found efficiently. Note
that the condition that no connected component is isomorphic to a
short cycle is necessary (a disjoint union of short cycles has average degree 2
but no vertex of degree 1 and no long chain of vertices of degree 2).

\begin{lemma}\label{lem:mad2}
  Let $G$ be an $n$-vertex graph with $\ad(G)\le 2+\tfrac1{40q}$ and
  without component isomorphic to a cycle of length less than $2q+2$, for some
  $q\ge 1$. Let $S$ be the set of vertices of degree at most 1 in $G$,
  and let $P$ be the set of vertices belonging to a path consisting of
  at least $2q+1$ vertices, all
  of degree 2 in $G$ (in particular each vertex of $P$ has degree
  2 in $G$). Then  $|S\cup P|\ge \tfrac1{40q}\, n$.
\end{lemma}

\begin{proof}
Let $\epsilon=\tfrac1{40q}$. For $i=0,1,2$, let $V_i$ be the set of vertices of degree $i$, and let
$V_3^+$ be the set of vertices of degree at
least 3. We denote by $n_0$, $n_1$, $n_2$, and $n_3^+$ the cardinality
of these four sets. 
Since $G$ has average degree at most $2+\epsilon$, we
have \[n_1+2n_2+ 3n_3^+\le n_1+2n_2+\sum_{v\in V_3^+}d_G(v)\le
(2+\epsilon)(n_0+n_1+n_2+n_3^+),\] and thus $n_3^+\le
\tfrac{2+\epsilon}{1-\epsilon} \cdot
n_0+\tfrac{1+\epsilon}{1-\epsilon}\cdot
n_1+\tfrac{\epsilon}{1-\epsilon} \cdot n_2\le \tfrac52 \,n_0+\tfrac32\,n_1+\tfrac32\epsilon\,n_2$ (since $\epsilon\le \tfrac18$).

\smallskip

Let $H$ be the multigraph obtained from $G$ by removing all connected
components isomorphic to a cycle, and then replacing each
maximal path of vertices of degree 2 in $G$ by a single edge (i.e., for
each maximal path $P$ in which all internal vertices have degree 2 in
$G$, we delete these internal vertices and add an edge between the two 
endpoints of $P$). Note
that $H$ has no vertices of degree 2, and it contains
precisely $n_0+n_1+n_3^+$ vertices. Observe also that the number $m_H$ of edges
of $H$ is precisely $\tfrac12 \sum_{v\in V_0\cup V_1 \cup V_3^+}
d_G(v)$. It thus follows from the inequalities above that
\begin{eqnarray*}
  m_H &\le &\tfrac12(2+\epsilon)(n_0+n_1+n_2+n_3^+)-n_2\\
  & = & (1+\tfrac{\epsilon}2)(n_0+n_1+n_3^+)+\tfrac{\epsilon}2 \cdot
        n_2\\
    &\le & \tfrac54\,n_0+\tfrac54\,n_1+\tfrac54(\tfrac52\,n_0+\tfrac32\,n_1+\tfrac32\epsilon\,n_2)+\tfrac{\epsilon}2 \cdot
           n_2.\\
  &\le & 5(n_0+n_1)+3\epsilon\, n_2.
\end{eqnarray*}

We now set $S=V_0\cup V_1$. Thus,
if $|S|=n_0+n_1\ge \tfrac1{40q}\,n= \epsilon \,n$ we have sets $S$ and
$P=\emptyset$ satisfying all required properties. Hence, we
can assume in the remainder of the proof that
\begin{eqnarray*}
  n_0+n_1 \le \epsilon\,n & \le &  \epsilon \,(n_0+n_1+n_2+n_3^+)\\
  & \le & \epsilon \,(n_0+n_1+n_2+\tfrac52
          (n_0+n_1)+\tfrac32\epsilon\,n_2)\\
  & \le & (n_0+n_1)\tfrac72\,\epsilon +
          n_2\,\epsilon(1+\tfrac32\,\epsilon)\\
          & \le & (n_0+n_1)\tfrac72\,\epsilon +2
          n_2\,\epsilon
\end{eqnarray*}
It follows that $$n_0+n_1\le \frac{1}{1-\tfrac72\,\epsilon} \cdot 2
          n_2\,\epsilon\le 3\epsilon \,n_2,$$ since $\epsilon\le
          \tfrac1{12}$. 
This implies that $n_3^+\le 9\epsilon\, n_2$ since otherwise the average
degree would be larger than $(2+\epsilon)$.  Consequently, we obtain $n\le (1+12\epsilon)\,n_2$, and $m_H\le 18\epsilon \,n_2$.

          \smallskip

In $G$,
remove all the vertices of degree at most 1 and at least 3. We are
left with $m_H$ (possibly empty) paths $P_1,P_2, \ldots, P_{m_H}$ of
vertices of degree 2 in $G$, each corresponding to an edge of $H$
(each edge of $H$ is either a path in $G$ of vertices of degree two, or a real edge of $G$ in
which case the corresponding path is empty),
plus a certain number of cycles (consisting of vertices of degree 2 in
$G$). Since $G$ has no connected component isomorphic to a cycle of
length less than $2q+2$, each vertex of such a cycle
is included in a path consisting of at least $2q+1$ vertices of degree
2 in $G$, so all these vertices can be added to the set $P$. We also
add to $P$ all the paths $P_i$ ($1\le i \le m_H$) containing at least
$2q+1$ vertices. As a consequence, the set $P$ contains all the vertices of
degree 2 in $G$, except those which only belong to
paths $P_i$ of at most $2q$ vertices.
So we have
$|P|\ge n_2-2qm_H$.  
By the inequalities above, we
have $$n_2-2qm_H \ge n_2(1-18\epsilon \cdot 2q) \ge \frac{1-2\cdot 18\epsilon q}{1+12\epsilon}\cdot n\ge
\tfrac{n}{40},$$ where the last inequality
follows from $\epsilon= \tfrac1{40q}$. Since $\tfrac{n}{40}\ge
\epsilon\,n$, the set $P$ contains at least $\epsilon\,n$ vertices, as desired.
\end{proof}

The \emph{girth} of a graph $G$ is the length of a shortest cycle in
$G$ (if the graph is acyclic we set its girth to $+\infty$).
We recall that the \emph{maximum average degree} of a graph $G$, denoted by
$\mad(G)$, is the maximum of the average degrees of the
subgraphs of $G$.
We now explain how to apply  Lemma~\ref{lem:mad2} to design a
distributed algorithm for $(2q+1\!:\!q)$-coloring. Note that
Theorem~\ref{thm:distrmad2}, which we restate below for convenience,  requires that the
\emph{maximum} average degree of $G$ is close to 2 and the girth is at
least $2q+2$, while the previous
result only required that the average degree is close to 2, and there
is no component isomorphic to a cycle of length less than $2q+2$. As
observed by a reviewer, there are similarities between our approach and the \emph{rake and
compress} technique of Miller and Reif~\cite{MR89} (see
also~\cite{B+S19} and~\cite{CP19} for distributed algorithms in trees
using this technique).



\thmdistmad*

\begin{proof}
  The algorithm proceeds similarly as
in~\cite{GPS88}. We set $G_0:=G$ and for $i=1$ to $\ell=O(\log n)$ we
define $S_{i-1}\cup P_{i-1}$ as the set of vertices of degree at most 2 given by
applying Lemma~\ref{lem:mad2} to $G_{i-1}$ (which has average degree
at most $2+\tfrac1{40q}$ since $\mad(G)\le 2+\tfrac1{40q}$), and set
$G_i:= G_{i-1}-(S_{i-1}\cup P_{i-1})$. Note that each $S_{i}\cup P_{i}$ consists of a set of
vertices of $V(G_i)$ of
size at least $\tfrac1{40q}|V(G_i)|$, and in
particular we can choose $\ell=O(\log n)$ such that $G_\ell$ is
empty. Note that the induced subgraph $G[S_i\cup
P_i]=G_i[S_i\cup P_i]$ consists of isolated vertices and edges, paths consisting of at least $2q+1$ vertices, all of degree
2 in
$G_i$, and cycles consisting of at least $2q+2$ vertices, all of degree
2 in $G_i$. 

Note that each $S_i\cup P_i$ can be computed in $O(q)$ rounds (each
vertex only needs to look at its neighborhood at distance at most
$2q+1$), and thus the decomposition of $G$ into
$S_1,P_1,\ldots,S_\ell,P_\ell$ (and the sequence of graphs $G_1,\ldots,G_\ell$)
can be computed in $O(q\log n)$ rounds.

For each $1\le i \le \ell$, \emph{in parallel}, compute a
maximal set $I_i$ of vertices at pairwise distance at least $2q+1$ in $G[S_i\cup
P_i]$. Recall that the vertices of $S_i$ have degree
at most $1$ in $G_i$, so they induce isolated vertices or isolated
edges in $G_i$ (and $G$), while $P_i$ induces a disjoint union of cycles of length at least
$2q+2$ and paths of at least $2q+1$ vertices, each consisting only of vertices
of degree 2 in $G_i$. In particular, by maximality of $I_i$, the set $P_i-I_i$ induces a collection of
disjoint (and pairwise non-adjacent)
paths of at least $2q+1$ and at most
$4q+2$ vertices (except the first and last segment of each path of $P_i$, which
might contain fewer vertices). For each path of $P_i$, discard from $I_i$
the first and last vertex of $I_i$ in the path (these two vertices might
coincide if a path of $P_i$ contains a single vertex of $I_i$), and call $I_i'$ the
resulting subset of $I_i$ (note that each vertex $x\in I_i$ can check in
$O(q)$ rounds if it belongs to $I_i'$ by inspecting the lengths of the
two subpaths of vertices of degree $2$ adjacent to $x$, if any of them is smaller than $2q+1$
then $x\in I_i'$). By maximality of $I_i$ and the definition
of $I_ i'$, the set $P_i-I_i'$ induces a collection of
disjoint (and pairwise non-adjacent)
paths of length at least $2q+1$ and at most
$8q+2$.
Note that each graph $G[S_i\cup
P_i]$ has maximum degree at most $2$, so a maximal set $I_i$ of
vertices at pairwise distance at least $2q+1$ can be computed in $O(q+\log^* n)$
rounds in the $(2q+1)$-th power of $G[S_i\cup
P_i]$~\cite{BEK14}, and thus in $O(q^2+q\log^* n)$
rounds in $G$. Since the computation of the sets $I_i$ is made in
parallel in each $G[S_i\cup
P_i]$, this step takes $O(q^2+\log^* n)$ rounds.

We now color each $S_i\cup P_i$ in reverse order, i.e., from
$i=\ell-1$ to 0. For the components induced by $S_i$ this can be
done greedily, since the vertices have degree at most 1 in $G_i$, they have at most one
colored neighbor and thus at most $q$ forbidden colors (and at least
$q+1$ available colors). For the components induced by $P_i$, we start
by coloring $I_i'$ arbitrarily, and then extend the coloring greedily
to $P_i-I_i'$
until each path of uncolored vertices has
size precisely $2q+1$ (this can be done in $O(q)$ rounds). We then use Theorem~\ref{thm:pathchoos} (each
endpoint of an uncolored path has a list of at least $2q+1-q\ge q+1$
available colors).
Each coloring extension takes $O(q)$ rounds, so overall this part takes
$O(q\log n)$ rounds. It follows that the
overall round complexity is $O(q^2 + q\log n)$, as desired.
\end{proof}

This immediately implies the following.

\begin{corollary}\label{cor:tree}
Let $G$ be an $n$-vertex tree. Then for any fixed $q\ge 1$, a $(2q+1\!:\!q)$-coloring of $G$ can be computed in
$O(q\log n+q^2)$ rounds.
\end{corollary}

Note that given any $(2q+1\!:\!q)$-coloring, we can deduce a $(q+2)$-coloring
in a single round (each vertex chooses the smallest color in its set
of $q$ colors given by the $(2q+1\!:\!q)$-coloring), while coloring trees with a constant number of colors
takes $\Omega(\log n)$ rounds~\cite{L92}, so the round complexity in
Corollary~\ref{cor:tree} is best possible.

\medskip

For $k\ge 1$, a graph $G$ is \emph{$k$-path-degenerate} if any non-empty subgraph $H$ of $G$
contains a vertex of degree at most 1, or a path consisting of $k$
vertices of degree 2 in $H$.

\begin{lemma}\label{lem:ppd}
If $G$ is $k$-path-degenerate, then $\mad(G)\le 2+\tfrac2k$.
\end{lemma}

\begin{proof}
Let $H$ be a subgraph of $G$. Let $n$ and $m$ be the number of
vertices and edges of $H$. We prove that $\ad(H)=2m/n\le 2+\tfrac2k$ by
induction on $n$. If $H$ is empty, then the result
is trivial, so assume that $n\ge 1$. Since $G$ is
$k$-path-degenerate, $H$ contains a vertex of degree at most 1 or a
path of $k$ vertices of degree 2 in $H$. Assume first that $H$
contains a vertex $v$ of degree at most 1. Then $H-v$ contains $n-1$
vertices and at least $m-1$ edges, and thus by induction
$2+\tfrac2k\ge \ad(H-v)\ge \tfrac{2m-2}{n-1}$. It follows that $2m\le
(n-1)(2+\tfrac2k)+2\le n (2+\tfrac2k)$, and thus $\ad(H)\le
2+\tfrac2k$, as desired. Assume now that $H$ contains a path $P$ of
$k$ vertices of degree 2 in $H$. Then $H-P$ contains $n-k$ vertices
and $m-k-1$ edges, and by induction $2+\tfrac2k\ge \ad(H-P)=
\tfrac{2m-2k-2}{n-k}$. It follows that $2m\le
(n-k)(2+\tfrac2k)+2k+2\le n(2+\tfrac2k)$, thus $\ad(H)\le
2+\tfrac2k$, as desired.
\end{proof}

It was proved by Gallucio, Goddyn, and Hell~\cite{GGH01} that if
$\mathcal{C}$ is a proper minor-closed class, or a class closed under
taking topological minors, then for any $k\ge 1$ there is a girth
$g(k)$ such that any graph $G\in \mathcal{C}$ with girth at least
$g(k)$ is $k$-path-degenerate. Using Lemma~\ref{lem:ppd} and
Theorem~\ref{thm:distrmad2}, this immediately implies the following.

\begin{corollary}\label{cor:minor}
  For any integer $q\ge 1$  and any proper class $\mathcal{C}$ that is closed under taking
  minors or topological minors, there is an integer $g$ such
that any $n$-vertex graph $G\in \mathcal{C}$ of girth at least $g$ can be $(2q+1\!:\!q)$-colored in $O(\log n)$ rounds.
\end{corollary}



\begin{acknowledgement} 
We thank the anonymous reviewers for their detailed
comments and suggestions.
\end{acknowledgement}

\end{document}